\documentclass[pra,aps,twocolumn,superscriptaddress]{revtex4-1}
\usepackage[dvips]{graphicx}
\usepackage{amsmath,amssymb,amsthm,dsfont,fullpage,hyperref}

\newcommand{\bra}[1]{\langle#1|}
\newcommand{\ket}[1]{|#1\rangle}
\newcommand{\braket}[2]{\langle#1|#2\rangle}
\newcommand{\Tr}{\operatorname{Tr}}

\newcommand{\hilb}[1]{\mathcal{#1}}

\newtheorem{dfn}{Definition}
\newtheorem{lmm}{Lemma}
\newtheorem{thm}{Theorem}
\newtheorem{cor}{Corollary}

\begin{document}

\title{Accessible Information and Informational Power of Quantum $2$-designs}

\author{Michele \surname{Dall'Arno}}

\email{cqtmda@nus.edu.sg}

\affiliation{Centre for Quantum Technologies, National University of Singapore,
  3 Science Drive 2, 117543 Singapore, Republic of Singapore}

\date{\today}

\begin{abstract}
  The accessible information and the informational power quantify the amount of
  information extractable from a quantum ensemble and by a quantum measurement,
  respectively. So-called spherical quantum $2$-designs constitute a class of
  ensembles and measurements relevant in testing entropic uncertainty relations,
  quantum cryptography, and quantum tomography. We provide lower bounds on the
  entropy of $2$-design ensembles and measurements, from which upper bounds on
  their accessible information and informational power follow, as a function of
  the dimension only. We show that the statistics generated by $2$-designs,
  although optimal for the abovementioned protocols, never contains more than
  one bit of information. Finally, we specialize our results to the relevant
  cases of symmetric informationally complete (SIC) sets and maximal sets of
  mutually unbiased bases (MUBs), and we generalize them to the arbitrary-rank
  case.
\end{abstract}

\maketitle

\section{Introduction}

Quantum theory is arguably the most complete and successful description of the
inherent evolution of physical systems. At the same time, any information we can
access about physical systems is ultimately classical. The interface between the
quantum and classical domains is constituted by quantum ensembles and quantum
measurements, the former capable of preparing quantum states upon input of some
classical information, the latter of extracting some classical information from
quantum systems.

From the fundamental viewpoint, it is thus crucial to characterize ensembles and
measurements in terms of the maximal amount of extractable information, leading
to results such as entropic uncertainty relations~\cite{San95, BW07, WW10, BR11,
  BHOW13} and device-independent quantum information
processing~\cite{Bru14}. This characterization has deep consequences in a
plethora of applications, such as information locking~\cite{DHLST04}, quantum
cryptography~\cite{BB84}, tomography~\cite{DDPPS02}, communication~\cite{Dal11},
witnessing~\cite{GBHA10, DPGA12}, private decoupling~\cite{BGK09, Bus09},
purification of noisy measurements~\cite{DDS10}, and error
correction~\cite{Bus08, CDDMP10}, where the efficiency -- or the success itself
-- of the protocol depends on the generated input-output statistics.

In this paper we address the problem of quantifying the maximal amount of
information that can be extracted from a quantum ensemble, or by a quantum
measurement. The former problem, known as the {\em accessible information
  problem}, was introduced~\cite{LL66, Hol73, Bel75a, Bel75b, Dav78, JRW94}
almost half a century ago, while the latter, known as the {\em informational
  power problem}, is much younger~\cite{DDS11, OCMB11, Hol12, Hol13, SS14,
  DBO14, Szy14}.

We will focus on a relevant class of ensembles and measurements, known as
spherical quantum $2$-designs~\cite{AE07}, whose distinctive feature is to share
many properties with the uniform distribution. Relevant examples of $2$-designs
are so-called symmetric, informationally complete (SIC)
measurements~\cite{Zau99, SG10}, and maximal sets of mutually unbiased bases
(MUBs)~\cite{KR05, BWB10}. They play a crucial role in our understanding of the
state space~\cite{FS03, Fuc04}, in quantum Bayesianism~\cite{FS09, FS11, AEF11,
  Fuc12}, and are key ingredients in many of the abovementioned quantum
information processing protocols.

We derive lower bounds on the entropy of the input distribution of $2$-design
ensembles and of the output distribution of $2$-design measurements. From these
results, we derive upper bounds on the accessible information and the
informational power of $2$-design ensembles and measurements, as a function of
the dimension of the system only. As a consequence, we show that, perhaps
surprisingly, the statistics generated by $2$-designs, although optimal for the
abovementioned protocols, never contains more than one bit of information. As
particular cases, we provide the accessible information and informational power
of SIC and MUB ensembles and POVMs (analytically for dimensions two and three,
numerically otherwise). Finally, we extend our results to generalizations of
SICs and MUBs with arbitrary rank.

The paper is structured as follows. Preliminary concepts are summarized in
Sec.~\ref{sec:formalism}. Lower bounds on the entropy of $2$-designs are derived
in Sec.~\ref{sec:entropybound}, from which bounds on accessible information and
informational power of $2$-designs are derived in Sec.~\ref{sec:infobound}. Our
results are specialized to the cases of SICs and MUBs in Sec.~\ref{sec:sicmub},
and generalized to the arbitrary-rank case in Sec.~\ref{sec:gensicmub}. We
conclude the paper discussing some open problems in Sec.~\ref{sec:conclusion}.

\section{Formalism}
\label{sec:formalism}

Let us recall some basic facts~\cite{NC00} from quantum information theory. Any
quantum system is associated with an Hilbert space $\hilb{H}$, and we denote
with $L(\hilb{H})$ the space of linear operators on $\hilb{H}$. We consider only
finite dimensional Hilbert spaces. A {\em quantum state} $\rho$ is a positive
semidefinite operator in $L(\hilb{H})$ such that $\Tr[\rho] \le 1$. Any
preparation of a quantum system is described by an {\em ensemble}, namely an
operator-valued measurable function $E = \{ \rho_x \}$ from real numbers $x$ to
states $\rho_x \in L(\hilb{H})$, such that $\sum_x\Tr[\rho_x] = 1$.

A {\em quantum effect} $\Pi$ is a positive semidefinite operator in
$L(\hilb{H})$ such that $\Pi \le \openone$. Any measurement on a quantum system
is described by a {\em positive-operator valued measure} (POVM), namely an
operator-valued measurable function $P = \{ \Pi_y \}$ from real numbers $y$ to
effects $\Pi_y \in L(\hilb{H})$, such that $\sum_y \Pi_y = \openone$, where
$\openone$ denotes the identity operator. Given an ensemble $E = \{ \rho_x \}$
and a POVM $P = \{ \Pi_y \}$, the joint probability $p_{x,y}$ of state $\rho_x$
and outcome $\Pi_y$ is given by the Born rule, namely $p_{x,y} = \Tr[\rho_x
  \Pi_y]$.

\subsection{Accessible information and informational power}

Let us recall some basic definitions~\cite{Cov06} from classical information
theory. A random variable $X$ is a function that maps from its domain, the
sample space, to its range, the real numbers, according to a probability
distribution $p_x$. Given a random variable $X$, its {\em Shannon entropy}
$H(X)$ defined as
\begin{align*}
  H(X) := - \sum_x p_x \log p_x,
\end{align*}
is a measure of the lack of information about the outcome of $X$. We write
$\log$ for binary logarithms and $\ln$ for natural logarithms, and we express
informational quantities in bits.

Given two random variables $X$ and $Y$, their {\em joint Shannon entropy}
$H(X,Y)$ is defined as
\begin{align*}
  H(X,Y) := - \sum_{x,y} p_{x,y} \log p_{x,y},
\end{align*}
is a measure of the lack of information about the joint outcomes of $X$ and
$Y$. The {\em conditional Shannon entropy} $H(Y|X)$ defined as $H(Y|X) := H(X,Y)
- H(X)$ is a measure of the lack of information about the outcome of $X$ given
the knowledge of the outcome of $Y$. The {\em mutual information} $I(X;Y)$
defined as $I(X;Y) := H(X) + H(Y) - H(X,Y) = H(Y) - H(Y|X) = H(X) - H(X|Y)$ is a
measure of of how much information about each random variable is carried by the
other one. Given an ensemble $E = \{ \rho_x \}$ and a POVM $P = \{ \Pi_y \}$, we
denote with $I(E,P)$ the mutual information $I(X;Y)$ between random variables
$X$ and $Y$ distributed according to $p_{x,y} = \Tr[\rho_x \Pi_y]$.

The accessible information~\cite{LL66, Hol73, Bel75a, Bel75b} is a measure of
how much information can be extracted from an ensemble.

\begin{dfn}[Accessible information]\label{def:accinfo}
  The accessible information $A(E)$ of an ensemble $E$ is the supremum over any
  POVM $P$ of the mutual information $I(E, P)$, namely
  \begin{align*}
    A(E) := \sup_{P} I(E, P).
  \end{align*}
\end{dfn}

The informational power~\cite{DDS11} is a measure of how much information can be
extracted by a POVM.

\begin{dfn}[Informational power]\label{def:infopower}
  The informational power $W(P)$ of a POVM $P$ is the supremum over any ensemble
  $E$ of the mutual information $I(E, P)$, namely
  \begin{align*}
    W(P) := \sup_{E} I(E, P).
  \end{align*}
\end{dfn}

We recall some results about accessible information and informational power that
will be useful in the following.

\begin{lmm}\label{lmm:duality}
  For any POVM $P = \{ \Pi_y \}$, the informational power $W(P)$ is the supremum
  over normalized states $\rho$ of the accessible information of the ensemble
  $\{ \rho^{1/2} \Pi_y \rho^{1/2} \}$, namely
  \begin{align}\label{eq:duality}
     W(P) = \sup_{\rho} A(\{ \rho^{1/2} \Pi_y \rho^{1/2} \}).
  \end{align}
\end{lmm}

\begin{proof}
  See Refs.~\cite{DDS11, DBO14}.
\end{proof}

\begin{lmm}\label{lmm:scrooge}
  For any ensemble $E$ of pure states and for any POVM $P$ with
  rank-one elements, the accessible information $A(E)$ and the
  informational power $W(P)$ are bounded as follows:
  \begin{align}\label{eq:scrooge}
    0 \le A(E) \le \log d.\\
    \log d - \frac1{\ln 2} \sum_{n=2}^d \frac1n \le W(P) \le
    \log d.
  \end{align}
\end{lmm}

\begin{proof}
  See Refs.~\cite{Hol73, JRW94, DBO14}.
\end{proof}

\subsection{Spherical quantum $2$-designs}

A spherical quantum $t$-design is a discrete probability distribution over
quantum states that shares some properties with the uniform distribution.

\begin{dfn}[Spherical quantum $t$-design]\label{def:2design}
  A spherical quantum $t$-design $\{ p_x, \ket{\phi_x} \}_{x=1}^N$ is a
  probability distribution $p_x$ over normalized pure states $\ket{\phi_x}$ such
  that
  \begin{align}\label{eq:2design}
    \sum_{x=1}^N p_x \left( \ket{\phi_x}\bra{\phi_x} \right)^{\otimes
      s} = \int \left( \ket{\phi}\bra{\phi} \right)^{\otimes s} d \phi
  \end{align}
  holds for any $s \le t$, where the integral is over the Haar
  measure.
\end{dfn}

We recall some results about quantum $t$-designs that will be useful in the
following.

\begin{lmm}\label{lmm:avgstate}
  The integral over Haar measure in Eq.~\eqref{eq:2design} is given by
  \begin{align}\label{eq:avgstate}
    \int \left( \ket{\phi}\bra{\phi} \right)^{\otimes s} d \phi =
    \frac1M P_{\textrm{sym}},
  \end{align}
  where $P_{\textrm{sym}}$ is the projector over the symmetric
  subspace and $M = \binom{s+d-1}{s}$ is its norm.
\end{lmm}

\begin{proof}
  See Ref.~\cite{AE07}.
\end{proof}

A $t$-design is called uniformly distributed (or unweighted) if $p_x = 1/N$ for
all $x$. In this work we will focus on $2$-designs. An ensemble $E = \{ \rho_x
\}$ of pure states is a {\em $2$-design ensemble} if $\{ p_x, \ket{\phi_x} \}$
is a $2$-design, with $p_x := \Tr[\rho_x]$ and $\ket{\phi_x} \bra{\phi_x} :=
\rho_x / \Tr[\rho_x]$. Notice that upon setting $s = 1$ in
Eqs.~\eqref{eq:2design} and~\eqref{eq:avgstate}, it follows that the average
state of any $2$-design ensemble is $\openone/d$. This allows us to define
$2$-design POVMs as follows. A POVM $P = \{ \Pi_y \}$ with rank-one elements is
a {\em $2$-design POVM} if $\{ q_y, \ket{\pi_y} \}$ is a $2$-design, with $q_y
:= \Tr[\Pi_y]/d$ and $\ket{\pi_y} \bra{\pi_y} := \Pi_y / \Tr[\Pi_y]$.

Noticeable examples of uniformly-distributed $2$-designs are symmetric
informationally complete (SIC) sets~\cite{Zau99}, for which $N = d^2$, and $d+1$
mutually unbiased bases (MUBs)~\cite{KR05}, for which $N = d(d+1)$.

\begin{dfn}[SIC]\label{dfn:sic}
  A $d$-dimensional SIC set $\{ p_x, \ket{\phi_x} \}_{x=1}^{d^2}$ is a uniform
  probability distribution $p_x = 1/d^2$ over normalized pure states
  $\ket{\phi_x}$ such that
  \begin{align}\label{eq:sic}
    |\braket{\phi_x}{\phi_y}|^2 = \frac{d\delta_{x,y}+1}{d+1}.
  \end{align}
\end{dfn}

\begin{dfn}[MUB]\label{dfn:mub}
  A $d$-dimensional $(d+1)$-MUB $\{ p_{b,x}, \ket{\phi_{b,x}} \}_{b=1,
    x=1}^{d+1,d}$ is a uniform probability distribution $p_{b,x} = 1/(d(d+1))$
  over normalized pure states $\ket{\phi_{b, x}}$ such that
  \begin{align}\label{eq:mub}
    |\braket{\phi_{b,x}}{\phi_{b',x'}}|^2 = \delta_{b,b'}
    \delta_{x,x'} + \frac1d (1 - \delta_{b,b'}).
  \end{align}
\end{dfn}

\section{Entropic bounds for $2$-designs}
\label{sec:entropybound}

In previous literature, entropic bounds for SICs and MUBs were
discussed~\cite{Ras13}. We generalize those results by providing bounds for
arbitrary uniformly-distributed $2$-designs.

\begin{thm}\label{thm:ensemblebound}
  The Shannon entropy $H(E, \Pi)$ of any $2$-design ensemble $E = \{ \rho_x
  \}_{x=1}^N$ uniformly distributed (namely, $\Tr[\rho_x] = 1/N$ for all $x$)
  with respect to effect $\Pi$ is bounded as follows:
  \begin{align}\label{eq:ensemblebound}
    H(E, \Pi) \ge \log \left( \frac{N(d+1)}{d} \frac{\Tr[\Pi]^2}{\Tr[\Pi]^2 +
      \Tr[\Pi^2]} \right).
  \end{align}
\end{thm}

\begin{proof}
  Let $\ket{\phi_x} \bra{\phi_x} := N \rho_x$. By setting $s = 2$ in
  Eqs.~\eqref{eq:2design} and~\eqref{eq:avgstate} it follows that
  \begin{align*}
    \sum_x \frac1N (\ket{\phi_x}\bra{\phi_x})^{\otimes 2} = \frac{\openone +
      S}{d(d+1)}.
  \end{align*}
  where we used the fact that the projector on the symmetric subspace of
  $\hilb{H}^{\otimes 2}$ is $P_{\textrm{sym}} = \frac12 (\openone + S)$ and $S$
  is the swap operator.
  
  Multiplying both sides by $\frac{d^2}{N} \frac{\Pi^{\otimes 2}}{\Tr[\Pi]^2}$
  and taking the trace we get
  \begin{align}\label{eq:coincidence1}
    \sum_x \frac{d^2}{N^2} \frac{(\bra{\phi_x} \Pi \ket{\phi_x})^2}{\Tr[\Pi]^2}
    = \frac{d}{N(d+1)} \frac{\Tr[\Pi]^2+\Tr[\Pi^2]}{\Tr[\Pi]^2},
  \end{align}
  where we used the fact that $\Tr[\Pi^{\otimes 2} S] = \Tr[\Pi^2]$ for any
  effect $\Pi$.

  Since the probability of state $\rho_x$ given effect $\Pi$ is given by
  \begin{align*}
    p_{x|\Pi} = \frac{d}{N} \frac{\bra{\phi_x} \Pi \ket{\phi_x}}{\Tr[\Pi]},
  \end{align*}
  the negative logarithm of the left-hand side of Eq.~\eqref{eq:coincidence1}
  can be upper bounded by means of Jensen's inequality as follows:
  \begin{align*}
    -\log \sum_x p_{x|\Pi}^2 \le -\sum_x p_{x|\Pi} \log p_{x|\Pi}.
  \end{align*}

  The right-hand side is the Shannon entropy $H(E, \Pi)$ of ensemble $E$ with
  respect to effect $\Pi$, so the statement follows.
\end{proof}

\begin{thm}\label{thm:povmbound}
  The Shannon entropy $H(P, \rho)$ of any $2$-design POVM $P = \{ \Pi_y
  \}_{y=1}^N$ uniformly distributed (namely, $\Tr[\Pi_y] = d/N$ for all $y$)
  with respect to state $\rho$ is bounded as follows:
  \begin{align}\label{eq:povmbound}
    H(P, \rho) \ge \log \left( \frac{N(d+1)}{d} \frac{\Tr[\rho]^2}{\Tr[\rho]^2 +
      \Tr[\rho^2]} \right).
  \end{align}
\end{thm}

\begin{proof}
  Let $\ket{\pi_y} \bra{\pi_y} := N/d \Pi_y$. By setting $s = 2$ in
  Eqs.~\eqref{eq:2design} and~\eqref{eq:avgstate} it follows that
  \begin{align*}
    \sum_y \frac1N (\ket{\pi_y}\bra{\pi_y})^{\otimes 2} = \frac{\openone +
      S}{d(d+1)}.
  \end{align*}
  where we used the fact that the projector on the symmetric subspace of
  $\hilb{H}^{\otimes 2}$ is $P_{\textrm{sym}} = \frac12 (\openone + S)$ and $S$
  is the swap operator.
  
  Multiplying both sides by $\frac{d^2}{N} \frac{\rho^{\otimes 2}}{\Tr[\rho]^2}$
  and taking the trace we get
  \begin{align}\label{eq:coincidence2}
    \sum_y \frac{d^2}{N^2} \frac{(\bra{\pi_y} \rho \ket{\pi_y})^2}{\Tr[\rho]^2}
    = \frac{d}{N(d+1)}\frac{\Tr[\rho]^2+\Tr[\rho^2]}{\Tr[\rho^2]},
  \end{align}
  where we used the fact that $\Tr[\rho^{\otimes 2} S] = \Tr[\rho^2]$ for any
  state $\rho$.

  Since the probability of outcome $\Pi_y$ given state $\rho$ is given by
  \begin{align*}
    q_{y|\rho} = \frac{q}{N} \frac{\bra{\pi_y} \rho \ket{\pi_y}}{\Tr[\rho]},
  \end{align*}
  the negative logarithm of the left-hand side of Eq.~\eqref{eq:coincidence2}
  can be upper bounded by means of Jensen's inequality as follows:
  \begin{align*}
   -\log \sum_y q_{y|\rho}^2 \le -\sum_y q_{y|\rho} \log q_{y|\rho}.
  \end{align*}

  The right-hand side is the Shannon entropy $H(P, \rho)$ of POVM $P$ with
  respect to state $\rho$, so the statement follows.
\end{proof}

By direct inspection it follows that state- and effect-dependent lower bounds in
Eqs.~\eqref{eq:ensemblebound} and~\eqref{eq:povmbound} are independent of the
norms of $\Pi$ and $\rho$, so without loss of generality we can set $\Tr[\rho] =
\Tr[\Pi] = 1$. Furthermore, those bounds can be made state- and
effect-independent by minimizing the right hand side of
Eqs.~\eqref{eq:ensemblebound} and~\eqref{eq:povmbound} over $\Pi$ and $\rho$,
respectively, with minimum achieved when $\Pi$ and $\rho$ are rank-one
(namely, $\Tr[\rho^2] = \Tr[\Pi^2] = 1$).

\section{Accessible information and informational power of $2$-designs}
\label{sec:infobound}

In previous literature, the accessible information and the informational power
of SICs were derived for dimension $2$~\cite{DDS11, OCMB11, SS14} and
$3$~\cite{DBO14, Szy14}, and tight bounds were provided for any
dimension~\cite{DBO14}. Bounds are also known for maximal sets of
MUBs~\cite{San95, DHLST04}. In this section we generalize those results by
providing bounds for $2$-designs in any dimension.

Upper bounds on the accessible information and informational power of
uniformly-distributed $2$-designs can be derived from
Theorems~\ref{thm:ensemblebound} and~\ref{thm:povmbound}. However, in order to
bound the accessible information of $2$-designs in
Theorem~\ref{thm:accinfobound} we use an alternative derivation which holds for
arbitrary (not necessarily uniformly-distributed) $2$-design ensembles.

\begin{thm}\label{thm:accinfobound}
  The accessible information $A(E)$ of any $d$-dimensional $2$-design ensemble
  $E = \{ \rho_x \}$ (not necessarily uniformly-distributed) is bounded as
  follows:
  \begin{align}\label{eq:accinfobound}
    A(E) \le \log \frac{2d}{d+1}.
  \end{align}
\end{thm}

\begin{proof}
  Let $p_x := \Tr[\rho_x]$ and $\ket{\phi_x} \bra{\phi_x} := \rho_x /
  \Tr[\rho_x]$. By setting $s = 2$ in Eqs.~\eqref{eq:2design}
  and~\eqref{eq:avgstate} it follows that
  \begin{align*}
    \sum_x p_x (\ket{\phi_x}\bra{\phi_x})^{\otimes 2} = \frac{\openone
      + S}{d(d+1)}.
  \end{align*}
  where we used the fact that the projector on the symmetric subspace
  of $\hilb{H}^{\otimes 2}$ is $P_{\textrm{sym}} = \frac12 (\openone +
  S)$ and $S$ is the swap operator.
  
  By Davies' theorem~\cite{Dav78} it suffices to optimize over POVMs with
  rank-one elements. Let $P = \{ \Pi_y \}$ be such a POVM and let $q_y :=
  \Tr[\Pi_y]/d$ and $\ket{\pi_y} \bra{\pi_y} := \Pi_y / \Tr[\Pi_y]$.
  Multiplying both sides by $d^2 \sum_y q_y (\ket{\pi_y}\bra{\pi_y})^{\otimes
    2}$ and taking the trace we get
  \begin{align}\label{eq:coincidence3}
    \sum_{x,y} p_x q_y d^2 |\braket{\phi_x}{\pi_y}|^4 = \frac{2d}{d+1}
  \end{align}
  where we used the fact that $\Tr[\Pi^{\otimes 2} S] = \Tr[\Pi^2]$ for
  any effect $\Pi$.

  Since the joint probability of state $\rho_x$ and outcome $\Pi_y$ is $p_{x,y}
  = p_x q_y d |\braket{\phi_x}{\pi_y}|^2$, and its marginals are $p_x$ and $q_y$
  (we recall that $\sum_x p_x \ket{\phi_x}\bra{\phi_x} = \openone/d$), the
  logarithm of the left-hand side of Eq.~\eqref{eq:coincidence3} can be lower
  bounded by means of Jensen's inequality as follows:
  \begin{align*}
    \sum_{x,y} p_{x,y} \log \frac{p_{x,y}}{p_x q_y} \le \log \left( \sum_{x,y}
    p_{x,y} \frac{p_{x,y}}{p_x q_y} \right).
  \end{align*}

  The left-hand side is the mutual information $I(E, P)$ between ensemble $E$
  and POVM $P$, namely
  \begin{align*}
    I(E, P) := \sum_{x,y} p_{x,y} \log \frac{p_{x,y}}{p_x q_y} \le \log
    \frac{2d}{d+1}.
  \end{align*}
  Since $A(E) := \sup_P I(E, P)$, the statement follows.
\end{proof}

\begin{thm}\label{thm:infopowerbound}
  The informational power $W(P)$ of any $2$-design POVM $P = \{ \Pi_y
  \}$ uniformly distributed (namely, $\Tr[\Pi_y] = d/N$ for all $y$)
  is bounded as follows:
  \begin{align}\label{eq:infopowerbound}
    W(P) \le \log \frac{2d}{d+1}.
  \end{align}
\end{thm}

\begin{proof}
  Let $\ket{\pi_y} \bra{\pi_y} := N/d \Pi_y$. By a Davies-like
  theorem~\cite{DDS11} it suffices to optimize over ensembles of pure
  states. Let $E = \{ \rho_x \}$ be such an ensemble and let $p_x :=
  \Tr[\rho_x]$ and $\ket{\phi_x} \bra{\phi_x} := \rho_x / \Tr[\rho_x]$. The
  joint probability of outcome $\Pi_y$ given state $\rho_x$ is then $p_{x,y} =
  p_x d/N |\braket{\phi_x}{\pi_y}|^2$.

  Let $X$, $Y$ be random variables with $X$ distributed according to $p_x$ and
  $Y$ such that $p(X=x, Y=y) = p_{x,y}$. Then one has
  \begin{align*}
    W(P) \ge I(E, P) = H(Y) - \sum_x H(Y|X=x)
  \end{align*}
  One can trivially upper bound $H(Y)$ as $H(Y) \le \log N$. Due to the
  state-independent version of Theorem~\ref{thm:povmbound}, we have that
  $H(Y|X=x) \ge \log \frac{N(d+1)}{2d}$, from which the statement follows.
\end{proof}

\section{The case of SICs and $(d+1)$-MUBs}
\label{sec:sicmub}

Almost forty years ago it was conjectured~\cite{Dav78}, and very recently
proved~\cite{DDS11, DBO14}, that the accessible information and the
informational power of $2$-dimensional SIC ensembles and POVMs (tetrahedral
configuration) is given by $\log 4/3$, with optimality achieved by the antipodal
tetrahedral configuration. Very recently, the accessible information and the
informational power $W(P)$ of $3$-dimensional SIC ensembles and POVMs were
proven~\cite{DBO14, Szy14} to be given by $\log\frac32$, with optimality
achieved by an orthonormal configuration. We notice that these results follows
as corollaries from Theorems~\ref{thm:accinfobound}
and~\ref{thm:infopowerbound}.

Very recently, it was also proven~\cite{SS14} that the accessible information
and the informational power of $2$-dimensional $3$-MUB ensembles and POVMs is
given by $1/3$, with optimality achieved by a $3$-MUB configuration. In this
Section we extend those results by deriving the accessible information and
informational power of $3$-dimensional $4$-MUB ensembles and POVMs.

\begin{cor}\label{thm:mubqutrit}
  The accessible information $A(E)$ of any $3$-dimensional $4$-MUB ensemble $E$
  and the informational power $W(P)$ of any $3$-dimensional $4$-MUB POVM $P$ is
  given by
  \begin{align}
    A(E) = W(P) = \log\frac32.
  \end{align}
  The POVM attaining $A(E)$ is a SIC POVM and the ensemble attaining $W(P)$ is a
  SIC ensemble.
\end{cor}

\begin{proof}
  The $3$-dimensional $4$-MUB ensemble $E$ or POVM $P$ is unique~\cite{BWB10} up
  to unitary transformations and permutations of elements. For some fixed
  orthonormal basis and up to a normalization, the coefficients of the vectors
  of $E$ and $P$ are given by the columns of the following matrix
  \begin{align*}
    \frac1{\sqrt3}
    \begin{bmatrix}
      \sqrt3 & 0 & 0 & 1 & 1 & 1 & 1 & 1 & 1 & 1 & 1 & 1\\
      0 & \sqrt3 & 0 & 1 & \omega & \omega^2 & \omega & \omega^2 & 1 & \omega^2 & \omega & 1\\
      0 & 0 & \sqrt3 & 1 & \omega^2 & \omega & \omega & 1 & \omega^2 & \omega^2 & 1 & \omega
    \end{bmatrix}
  \end{align*}
  where $\omega = e^{i 2\pi/3}$.

  Denote with $Q$ and $F$, respectively, the POVM and ensemble whose vectors are
  given, up to a normalization, by the columns of the following matrix
  \begin{align*}
    \frac1{\sqrt2}
    \begin{bmatrix}
      0 & 0 & 0 & 1 & 1 & 1 & -1 & \xi & \xi^*\\
      -1 & \xi & \xi^* & 0 & 0 & 0 & 1 & 1 & 1\\
      1 & 1 & 1 & -1 & \xi & \xi^* & 0 & 0 & 0
    \end{bmatrix}
  \end{align*}
  where $\xi = e^{i \pi/3}$. It is immediate to verify that $Q$ and $F$ are a
  SIC POVM and ensemble, respectively.

  By direct inspection it follows that $I(E, Q) = I(F, P) = \log3/2$. Since by
  Definitions~\ref{def:accinfo} and~\ref{def:infopower} we have $I(E, Q) \le
  A(E)$ and $I(F, P) \le W(P)$, and by Theorems~\ref{thm:accinfobound}
  and~\ref{thm:infopowerbound} we have that $A(E), W(P) \le \log3/2$, the
  statement follows.
\end{proof}

We now compare the optimal strategies attaining the accessible information and
the informational power, with the so-called pretty-good strategies~\cite{HW94,
  Hal97, Bus07, BH09}.  Given a $d$-dimensional $2$-design ensemble $E = \{
\rho_x \}$, its pretty-good POVM is $P = \{ \Pi_x \}$ with $\Pi_x = d \rho_x$;
analogously, given a $d$-dimensional $2$-design POVM $P = \{ \Pi_y \}$, its
pretty-good ensemble is $E = \{ \rho_y \}$ with $\rho_y = \Pi_y/d$.

For SIC ensembles and POVMs the mutual information given by the pretty-good
strategy is
\begin{align}\label{eq:prettygoodsic}
  I(E, d E) = I(P/d, P) = \log d - \frac{d-1}d \log(d+1).
\end{align}
The right-hand side of Eq.~\eqref{eq:prettygoodsic} is smaller than
the lower bound in Eq.~\eqref{eq:scrooge} for any $d$. Then the
pretty-good strategy for SIC ensembles and POVMs is suboptimal for any
$d$.

For $(d+1)$-MUB ensembles and POVMs the mutual information given by the
pretty-good strategy is
\begin{align}\label{eq:prettygoodmub}
  I(E, d E) = I(P/d, P) = \frac{\log d}{d+1}.
\end{align}
The right-hand side of Eq.~\eqref{eq:prettygoodmub} coincides with the
optimal value $1/3$ for $d = 2$. However, it is smaller than the value
$\log3/2$ provided by Corollary~\ref{thm:mubqutrit} for $d = 3$, and it is
smaller than the lower bound in Eq.~\eqref{eq:scrooge} for any $d \ge
4$. Then the pretty-good strategy for $(d+1)$-MUB ensembles and POVMs
is optimal for $d=2$ and suboptimal for $d \ge 3$.

We report the results of this Section in Fig.~\ref{fig:bounds}.
\begin{figure}[htb]
  \includegraphics[width=\columnwidth]{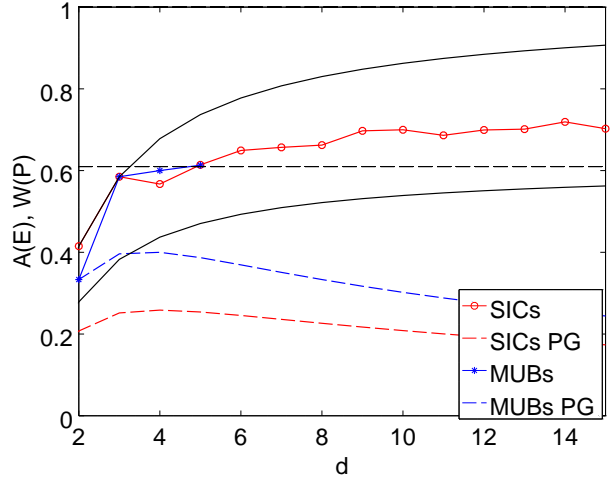}
  \caption{(Color online) Upper and lower bounds [thick (black) continuous
      lines] and their asymptotes [horizontal (black) dashed lines] on
    accessible information $A(E)$ and informational power $W(P)$ of any
    $2$-design ensemble $E$ and uniformly-distributed $2$-design POVM $P$, as a
    function of the dimension $d$, as given by Theorems~\ref{thm:accinfobound}
    and~\ref{thm:infopowerbound}. Accessible information and informational power
    of SIC ensembles and POVMs (red continuous line with circles), and
    corresponding pretty-good (PG) strategy [lower (red) dashed line] as in
    Eq.~\eqref{eq:prettygoodsic}. Accessible information and informational power
    of $(d+1)$-MUB ensembles and POVMs (blue continuous line with asterisks),
    and corresponding pretty-good (PG) strategy [upper (blue) dashed line] as in
    Eq.~\eqref{eq:prettygoodmub}. For red and blue continuous lines (with
    circles and asterisks, respectively), values are analytically derived for
    $d=2,3$ (see Corollary~\ref{thm:mubqutrit}), and numerically derived for $d
    \ge 4$.  Numerical optimization was performed over SICs (up to dimension
    $15$) and $(d+1)$-MUBs (up to dimension $5$, no example is known in
    dimension $6$) as provided in Refs.~\cite{SG10, BWB10}.}
  \label{fig:bounds}
\end{figure}

\section{Arbitrary-rank SICs and $(d+1)$-MUBs}
\label{sec:gensicmub}

Symmetric informationally complete sets and mutually unbiased bases were
generalized~\cite{KG13, KG14} to the case of arbitrary-rank states and elements
in the following way.

\begin{dfn}[Arbitrary-rank SIC set]\label{dfn:generalizedsic}
  A $d$-dimensional arbitrary-rank SIC set $\{ p_x, \rho_x \}_{x+1}^{d^2}$ is a
  uniform probability distribution $p_x = 1/d^2$ over mixed states $\rho_x$ such
  that
  \begin{align}\label{eq:generalizedsic}
    \Tr[\rho_x \rho_{x'}] = \delta_{x,x'} d^2 a + (1 - \delta_{x,x'})
    \frac{d(1-da)}{d^2-1}.
  \end{align}
  for some $1/d^3 \le a \le 1/d^2$ and $\sum_x \rho_x = d \openone$.
\end{dfn}

Notice that the rank-one case is recovered for $a = 1/d^2$, while for $a =
1/d^3$ one has a set of maximally mixed operators. An ensemble $E = \{ \sigma_x
\}$ is an {\em arbitrary-rank SIC ensemble} if $\{ p_x, \rho_x \}$ is an
arbitrary-rank SIC set, with $p_x := \Tr[\sigma_x]$ and $\rho_x := \sigma_x /
\Tr[\sigma_x]$. A POVM $P = \{ \Pi_y \}$ is an {\em arbitrary-rank SIC POVM} if
$\{ q_y, \rho_y \}$ is an arbitrary-rank SIC set, with $q_y := \Tr[\Pi_y]/d$ and
$\rho_y := \Pi_y / \Tr[\Pi_y]$.

\begin{dfn}[Arbitrary-rank $(d+1)$-MUB set]\label{dfn:generalizedmub}
  A $d$-dimensional arbitrary-rank $(d+1)$-MUB set $\{ p_{b, x}, \rho_{b, x}
  \}_{b = 1,\dots d+1, x=1,\dots d}$ is a uniform probability distribution $p_x
  = 1/d(d+1)$ over mixed states $\rho_{b, x}$ such that
  \begin{align}\label{eq:generalizedmub}
    & \Tr[\rho_{b, x} \rho_{b', x'}]\nonumber\\ = & \delta_{b,b'}\delta_{x,x'} k
    + \delta_{b,b'} (1 - \delta_{x,x'}) \frac{1-k}{d-1} +
    \frac{1-\delta_{b,b'}}d.
  \end{align}
  for some $1/d \le k \le 1$ and $\sum_{b, x} \rho_{b, x} = (d+1) \openone$.
\end{dfn}

Notice that the rank-one case is recovered for $k = 1$, while for $k = 1/d$ one
has a set of maximally mixed operators. An ensemble $E = \{ \sigma_{b, x} \}$ is
an {\em arbitrary-rank $(d+1)$-MUB ensemble} if $\{ p_{b, x}, \rho_{b, x} \}$ is
an arbitrary-rank $(d+1)$-MUB set, with $p_{b, x} := \Tr[\sigma_{b, x}]$ and
$\rho_{b, x} := \sigma_{b, x} / \Tr[\sigma_{b, x}]$. A POVM $P = \{ \Pi_y \}$ is
an {\em arbitrary-rank MUB POVM} if $\{ q_{b, y}, \rho_{b, y} \}$ is an
arbitrary-rank $(d+1)$-MUB set, with $q_{b, y} := \Tr[\Pi_{b, y}]/d$ and
$\rho_{b, y} := \Pi_{b, y} / \Tr[\Pi_{b, y}]$.

\begin{thm}\label{thm:generalizedsic}
  The accessible information $A(E)$ of any arbitrary-rank SIC ensemble $E = \{
  \rho_x \}$ and the informational power $W(P)$ of any arbitrary-rank SIC POVM
  $P = \{ \Pi_y \}$ are bounded as follows:
  \begin{align}
    A(E), W(P) \le \log \frac{d (d^2a + 1)}{d+1}.
  \end{align}
\end{thm}

\begin{proof}
  Let us first prove the statement for $W(P)$, then the statement for $A(E)$
  will immediately follows from Lemmas~\ref{lmm:duality} and~\ref{lmm:avgstate}.
  
  By a Davies-like theorem~\cite{DDS11} it suffices to optimize over ensembles
  of pure states. Let $F = \{ \sigma_x \}$ be such an ensemble and let $p_x :=
  \Tr[\sigma_x]$ and $\ket{\phi_x} \bra{\phi_x} := \sigma_x /
  \Tr[\sigma_x]$. The joint probability of state $\sigma_x$ and outcome $\Pi_y$
  is then $p_{x,y} = p_x \bra{\phi_x} \Pi_y \ket{\phi_x}$.

  Let $X$, $Y$ be random variables with $X$ distributed according to $p_x$ and
  $Y$ such that the joint probability of X and Y is $p_{x,y}$. Then one has
  \begin{align*}
    W(P) \ge I(E, P) = H(Y) - \sum_x p_x H(Y|X=x).
  \end{align*}
  One can trivially upper bound $H(Y)$ as $H(Y) \le \log d^2$. It was
  proven~\cite{Ras14} that
  \begin{align*}
    H(Y|X=x) \ge -\log \left( \frac{ad^3 -ad^2 +d -1}{d(d^2-1)} \right),
  \end{align*}
  from which the statement follows.
\end{proof}

\begin{thm}\label{thm:generalizedmub}
  The accessible information $A(E)$ of any generalized $(d+1)$-MUB ensemble $E =
  \{ \rho_x \}$ and the informational power $W(P)$ of any generalized
  $(d+1)$-MUB POVM $P = \{ \Pi_y \}$ are bounded as follows:
  \begin{align}
    A(E), W(P) \le \log \frac{d (k+1)}{d+1}.
  \end{align}
\end{thm}

\begin{proof}
  Let us first prove the statement for $W(P)$, then the statement for $A(E)$
  will immediately follows from Lemmas~\ref{lmm:duality} and~\ref{lmm:avgstate}.
  
  By a Davies-like theorem~\cite{DDS11} it suffices to optimize over ensembles
  of pure states. Let $F = \{ \sigma_x \}$ be such an ensemble and let $p_x :=
  \Tr[\sigma_x]$ and $\ket{\phi_x} \bra{\phi_x} := \sigma_x /
  \Tr[\sigma_x]$. The joint probability of state $\sigma_x$ and outcome $\Pi_y$
  is then $p_{x,y} = p_x \bra{\phi_x} \Pi_y \ket{\phi_x}$.

  Let $X$, $Y$ be random variables with $X$ distributed according to $p_x$ and
  $Y$ such that the joint probability of $X$ and $Y$ is $p_{x,y}$. Then one has
  \begin{align*}
    W(P) \ge I(E, P) = H(Y) - \sum_x p_x H(Y|X=x).
  \end{align*}
  One can trivially upper bound $H(Y)$ as $H(Y) \le \log (d(d+1))$. It was
  proven~\cite{CS14} that
  \begin{align*}
    H(Y|X=x) \ge \log \frac{(d+1)^2}{k+1},
  \end{align*}
  from which the statement follows.
\end{proof}

Notice that bounds on the accessible information and informational power of
rank-one SIC ensembles and POVMs~\cite{DBO14} can be obtained as a corollary of
Theorem~\ref{thm:generalizedsic} by setting $a = 1/d^2$; analogously, bounds for
$(d+1)$-MUBs~\cite{San95, DHLST04} can be obtained as a corollary of
Theorem~\ref{thm:generalizedmub} by setting $k = 1$. In the maximally mixed
case, by setting $a = 1/d^3$ (resp., $k = 1/d$) in
Theorem~\ref{thm:generalizedsic} (resp. Theorem~\ref{thm:generalizedmub}), one
has that accessible information and informational power vanish as expected.

\section{Conclusion and outlook}
\label{sec:conclusion}

We derived effect-dependent and effect-independent lower bounds on the entropy
of the input distribution of $2$-design ensembles; analogously, we derived
state-dependent and state-independent lower bounds on the entropy of the output
distribution of $2$-design measurements. From these results, we derived upper
bounds on the accessible information and the informational power of $2$-design
ensembles and measurements, as a function of the dimension of the system
only. As a consequence, we showed that, perhaps surprisingly, the statistics
generated by $2$-designs, although optimal for testing of entropic uncertainty
relations, quantum cryptography and tomography, never contains more than one bit
of information. As particular cases, we provided the accessible information and
informational power of SIC and MUB ensembles and POVMs (analytically for
dimensions two and three, numerically otherwise). Finally, we extended our
results to generalizations of SICs and MUBs with arbitrary rank.

We conclude by presenting a few relevant open problems. Analytically
characterizing the accessible information and the informational power of SIC
ensembles and POVMs in dimension larger than $3$ seems a hard task, as suggested
by the irregular behavior of the corresponding line in
Fig.~\ref{fig:bounds}. However, preliminary results seem to suggest that the
same task for $(d+1)$-MUB ensembles and POVMs could be feasible. Indeed,
consider the unique~\cite{BWB10} (up to unitary transformations and permutation
of elements) $4$-dimensional $(d+1)$-MUB ensemble $E$ and POVM $P$, whose
coefficients with respect to some fixed orthonormal basis are given (up to a
normalization) by the columns of the following matrices
\begin{align*}
  \frac12 &
  \begin{bmatrix}
    2 & 0 & 0 & 0 & 1 & 1 & 1 & 1 & 1 & 1\\
    0 & 2 & 0 & 0 & 1 & 1 &-1 &-1 &-1 &-1\\
    0 & 0 & 2 & 0 & 1 &-1 &-1 & 1 &-i & i\\
    0 & 0 & 0 & 2 & 1 &-1 & 1 &-1 &-i & i
  \end{bmatrix},\\
  \frac12 &
  \begin{bmatrix}
    & 1 & 1 & 1 & 1 & 1 & 1 & 1 & 1 & 1 & 1\\
    & 1 & 1 &-i &-i & i & i &-i &-i & i & i\\
    & i &-i &-i & i & i &-i &-1 & 1 &-1 & 1\\
    &-i & i &-1 & 1 &-1 & 1 &-i & i & i &-i
  \end{bmatrix}.
\end{align*}
The orthonormal POVM $Q$ and ensemble $F$ given by
\begin{align*}
  \begin{bmatrix}
    \frac1{\sqrt{2}} & \frac12 & 0 & \frac12\\
    \frac{i}{\sqrt{2}} & -\frac{i}2 & 0 & -\frac{i}2\\
    0 & -\frac{i}2 & \frac{i}{\sqrt{2}} & \frac{i}2\\
    0 & \frac12 & \frac1{\sqrt{2}} & -\frac12
  \end{bmatrix}
\end{align*}
are such that $I(E, Q) = I(F, P) = 3/5$, and we conjecture that this value is
optimal, namely $A(E) = W(P) = 3/5$.

Another relevant open problem is whether there exists a maximally informative
$2$-design, namely a $2$-design saturating the upper bound in
Theorems~\ref{thm:accinfobound} and~\ref{thm:infopowerbound} for any dimension
$d$. We showed that the answer is on the affirmative for $d = 2$, where SIC
ensembles and POVMs are optimal, and for $d = 3$, where both SIC and $(d+1)$-MUB
ensembles and POVMs are optimal. However, numerical results presented in
Fig.~\ref{fig:bounds} seem to suggest that, for $d \ge 4$, nor SICs nor
$(d+1)$-MUBs are maximally informative $2$-designs.

Finally, a very interesting open question is how to generalize arbitrary quantum
$2$-designs to the arbitrary-rank case, in the same spirit of the
generalizations for SICs and MUBs previously discussed~\cite{KG13, KG14}, and how
to quantify their accessible information and informational power. We believe
that these tantalizing open problems well deserve future investigation.

\section*{Acknowledgments}

The author is grateful to Francesco Buscemi, Chris Fuchs, Massimiliano
F. Sacchi, Wojciech S{\l}omczy\'{n}ski, Anna Szymusiak, and Vlatko Vedral for
very useful discussions, comments, and suggestions. This work was supported by
the Ministry of Education and the Ministry of Manpower (Singapore).


\begin{thebibliography}{}
\bibitem{San95} J. S\'anchez-Ruiz, Phys. Lett. A {\bf 201}, 125 (1995).
\bibitem{BW07} M. A. Ballester and S. Wehner, Phys. Rev. A {\bf 75},
  022319 (2007).
\bibitem{WW10} S. Wehner and A. Winter, New J. Phys. {\bf
  12}, 025009 (2010).
\bibitem{BR11} I. Bialynicki-Birula and L. Rudnicki, {\em Statistical
  Complexity: Applications in Electronic Structure}, Ed. K. D. Sen,
  (Springer, U.K., 2011), chapter 1.
\bibitem{BHOW13} F. Buscemi, M. J. W. Hall, M. Ozawa, and M. M. Wilde,
  Phys. Rev. Lett. {\bf 112}, 050401 (2014).

\bibitem{Bru14} N. Brunner, {\em Device-Independent Quantum Information
  Processing}, in Proceedings of the Quantum Information and Measurement
  Conference (2014).

\bibitem{DHLST04} D. P. DiVincenzo, M. Horodecki, D. W. Leung, J. A. Smolin, and
  B. M. Terhal, Phys. Rev. Lett. {\bf 92}, 067902 (2004).

\bibitem{BB84} C. H. Bennett and G. Brassard, {\em Quantum cryptography: Public
  key distribution and coin tossing}, in Proceedings of IEEE International
  Conference on Computers, Systems and Signal Processing 175, 8 (1984).
  
\bibitem{DDPPS02} G. Mauro D'Ariano, M. De Laurentis, M. G. A. Paris,
  A. Porzio, S. Solimeno, J. Opt. B: Quantum Semiclass. Opt. {\bf 4}, 127
  (2002).
    
\bibitem{Dal11} M. Dall'Arno, Scientifica Acta {\bf 5}, 22 (2011).

\bibitem{GBHA10} R. Gallego, N. Brunner, C. Hadley, and A. Ac\'in,
  Phys. Rev. Lett. {\bf 105}, 230501 (2010).
\bibitem{DPGA12} M. Dall'Arno, E. Passaro, R. Gallego, A. Ac\'in, Phys. Rev. A
  {\bf 86}, 042312 (2012).
  
\bibitem{BGK09} F. Buscemi, G. Gour, and J. S. Kim,
  Phys. Rev. A {\bf 80}, 012324 (2009).
\bibitem{Bus09} F. Buscemi, New J. Phys. {\bf 12}, 123002
  (2009).

\bibitem{DDS10} M. Dall'Arno, G. M. D'Ariano, and
  M. F. Sacchi, Phys. Rev. A {\bf 82}, 042315 (2010).

\bibitem{Bus08} F. Buscemi, Phys. Rev. A {\bf 77}, 012309
  (2008).
\bibitem{CDDMP10} G. Chiribella, M. Dall'Arno,
  G. M. D'Ariano, C. Macchiavello, and P. Perinotti,
  Phys. Rev. A {\bf 83}, 052305 (2011).

\bibitem{LL66} D. S. Lebedev, L. B. Levitin, Information and Control
  {\bf 9}, 1 (1966).
\bibitem{Hol73} A. S. Holevo, J. Multivariate Anal. {\bf 3}, 337
  (1973).
\bibitem{Bel75a} V. P. Belavkin, Stochastics {\bf 1}, 315
  (1975).
\bibitem{Bel75b} V. P. Belavkin, Radio Engineering and Electronic
  Physics {\bf 20}, 39 (1975).
\bibitem{Dav78} E. B. Davies, IEEE Trans. Inf. Theory {\bf
  24}, 596 (1978).
\bibitem{JRW94} R. Jozsa, D. Robb, and W. K Wootters, Phys. Rev. A
  {\bf 49}, 668 (1994).

\bibitem{DDS11} M. Dall'Arno, G. M. D'Ariano, and
  M. F. Sacchi, Phys. Rev. A {\bf 83}, 062304 (2011).
\bibitem{OCMB11} O. Oreshkov, J. Calsamiglia,
  R. Mu\~noz-Tapia, and E.  Bagan, New J. Phys. {\bf 13},
  073032 (2011).
\bibitem{Hol12} A. S. Holevo, Problems of Information
  Transmission {\bf 48}, 1 (2012).
\bibitem{Hol13} A. S. Holevo, Phys. Scr. {\bf 2013}, 014034 (2013).
\bibitem{SS14} W. S{\l}omczy\'{n}ski and A. Szymusiak,
  arXiv:1402.0375.
\bibitem{DBO14} M. Dall'Arno, F. Buscemi, and M. Ozawa, J. Phys. A:
  Math. Theor. {\bf 47}, 235302 (2014).
\bibitem{Szy14} A. Szymusiak, J. Phys. A: Math. Theor. {\bf 47}, 445301 (2014).

\bibitem{AE07} A. Ambainis and J. Emerson, {\em ``Quantum t-designs:
  t-wise independence in the quantum world}, in Proceedings of the
  Twenty-Second Annual IEEE Conference on Computational Complexity,
  129 (2007).
\bibitem{Zau99} G. Zauner, Quantendesigns – Grundzuge einer
  nichtkommutativen Designtheorie. Dissertation, Universitat Wien,
  1999.
\bibitem{SG10} A. J. Scott and M.  Grassl, Journal of Mathematical Physics
  {\bf 51}, 042203 (2010).
\bibitem{KR05} A. Klappenecker and M. Roetteler, Proceedings 2005 IEEE
  International Symposium on Information Theory (ISIT 2005), 1740 (2005).
\bibitem{BWB10} S. Brierley, S. Weigert, and I. Bengtsson, Quantum Info. \&
  Comp. {\bf 10}, 0803 (2010).

\bibitem{FS03} C. A. Fuchs and M. Sasaki, Quantum Inf. \& Comput. {\bf 3}, 377
  (2003).
\bibitem{Fuc04} C. A. Fuchs, Quantum Inf. \& Comput. {\bf 4}, 467 (2004).

\bibitem{FS09} C. A. Fuchs and R. Schack, Rev. Mod. Phys. {\bf 85}, 1693 (2013).
\bibitem{FS11} C. A. Fuchs and R. Schack, Foundations of Physics {\bf 41}, 345
  (2011).
\bibitem{AEF11} D. M. Appleby, \r{A}. Ericsson, and C. A. Fuchs, Foundations of
  Physics {\bf 41}, 564 (2011).
\bibitem{Fuc12} C. Fuchs, arXiv:1207.2141

\bibitem{NC00} I. L. Chuang and M. A. Nielsen, {\em Quantum Information and
  Communication} (Cambridge, Cambridge University Press, 2000).

\bibitem{Cov06} T. M. Cover and J. A. Thomas, {\em Elements of Information
  Theory} (Hoboken, Wiley-Interscience, 2006).

\bibitem{Ras13} A. E. Rastegin, Eur. Phys. J. D {\bf 67}, 269 (2013).

\bibitem{HW94} P. Hausladen and W. K. Wootters, J. Mod. Opt. {\bf 41}, 2385
  (1994).
\bibitem{Hal97} M. J. W. Hall, Phys. Rev. A {\bf 55}, 100 (1997).
\bibitem{Bus07} F. Buscemi, Phys. Rev. Lett. {\bf 99}, 180501 (2007)
\bibitem{BH09} F. Buscemi and M. Horodecki, OSID {\bf 16}, 29 (2009).

\bibitem{KG13} A. Kalev and G. Gour, J. Phys. A: Math. Theor. {\bf 47}, 335302
  (2014).
\bibitem{KG14} A. Kalev and G. Gour, New Journal of Physics {\bf 16}, 053038
  (2014).
\bibitem{Ras14} A. E. Rastegin, Physica Scripta {\bf 89}, 085101
  (2014).
\bibitem{CS14} B. Chen, and S.-H. Fei, arXiv:1407.6816.
\end{thebibliography}
\end{document}